\documentclass{amsart}

\usepackage{latexsym}
\usepackage{amsthm}
\usepackage{amssymb,amsmath}
\usepackage[dvips]{graphicx}
\usepackage[dvips]{graphics}
\usepackage{graphicx}
\usepackage{subfigure}
\usepackage{color}
\usepackage{epstopdf}
\usepackage{hyperref}
\usepackage{multirow}
\usepackage{enumerate}
\usepackage{extarrows}
\usepackage{pgfplots}
\usepackage{tikz-3dplot}
\usepackage[bottom]{footmisc}

\tdplotsetmaincoords{60}{115}
\pgfplotsset{compat=newest}

\newtheorem{thm}{Theorem}[section]
\newtheorem{lem}[thm]{Lemma}
\newtheorem{prop}[thm]{Proposition}

\theoremstyle{definition}

\newtheorem{eg}{Example}

\newtheorem{q}{Question}

\theoremstyle{remark}

\newcommand{\bP}{\mathbb P}

\newcommand{\bE}{\mathbb E}

\newcommand{\cW}{\mathcal W}

\newcommand{\Beta}{\operatorname{Beta}}

\numberwithin{equation}{section}

\begin{document}

\title{Discussion of  `A Gibbs sampler for a class of random convex polytopes'}


\author{Persi Diaconis}
\address{Department of Statistics, Stanford University, Stanford, CA, 94305}
\email{diaconis@math.stanford.edu}

\author{Guanyang Wang}
\address{Department of Statistics, Rutgers University, NJ 08854}
\email{guanyang.wang@rutgers.edu}


\maketitle
\pagestyle{empty}

\section{Introduction}
This very welcome paper by Jacob, Gong, Edlefsen and Dempster \cite{jacob2019gibbs} is a fascinating contribution to an important subject. Persistent questions about the lower and upper probabilities include:
\begin{enumerate}
	\item What does it do in simple problems?
	\item What does it do in real problems?
	\item What do the upper and lower probabilities mean?
	\item How can the needed computations be carried out?
\end{enumerate}

The paper offers answers to these questions by proposing a Gibbs sampler to perform statistical inference for categorical distributions using the Dempster-Shafer approach.  To be precise, let $\boldsymbol{x} = (x_i)_{i=1}^N$ be the observations, each comes from one of the $K$ categories. The model assumes that there exists a $\theta = (\theta_1, \cdots, \theta_K)$ in the $K-$simplex $\Delta :=\{(\theta_1, \theta_2, \cdots, \theta_K): \theta_i \geq 0 \text{~for every~} i, \sum_{i=1}^K \theta_i = 1\}$ such that $\bP(x_i = k) = \theta_k$ for every $i, k$. Moreover, it is  assumed that the observations $\boldsymbol{x}$ are generated from the following procedure.   
\begin{enumerate}
	\item For each $i \in \{1, \cdots, N\}$, sample $u_i$ uniformly from the simplex $\Delta$,
	\item Since $\theta$ partitions $\Delta$ into $K$ disjoint pieces, we set $x_i = k$ if $u_i \in \Delta_k(\theta)$, where $\Delta_k(\theta)$ is the ``subsimplex''
	which have the same vertices as $\Delta$ except for the $k$-th vertex replaced by $\theta$. 
\end{enumerate}

 Since this machinery generates $x_i$ through $u_i$, it is natural to ask what are the feasible points $\boldsymbol{u} \in  \{(u_1, u_2,\cdots u_n )\in \Delta^N\}$ that can generate the observed $\boldsymbol{x}$. For a fixed $\theta$, the set of feasible points $\mathcal R_{\boldsymbol{x}}(\theta)$ are easy to describe, as it is  the product of $N$ subsimplexes:
$$
\mathcal R_{\boldsymbol{x}}(\theta) = \prod_{i = 1}^N \Delta_{x_i}(\theta),
$$
and the whole feasible set is then 
$$\mathcal R_{\boldsymbol{x}} := \cup_{\theta\in \Delta}\mathcal R_{\boldsymbol{x}}(\theta).$$ 

From $u\in \mathcal R_{\boldsymbol{x}}$ one forms $\mathcal F_u := \{\theta\in \Delta: u_n \in \Delta_{x_n}(\theta) \text{~ for every~} n\}$ and assigns a lower probability to sets $\Sigma$ in the parameter space by the chance that $\mathcal F(u)$ is in $\Sigma$ (this last by Monte Carlo).

The paper's main contribution is a Gibbs sampler that samples from the uniform distribution of $R_{\boldsymbol{x}}$. The problem is challenging due to the complicated nature of  $R_{\boldsymbol{x}}$. In fact,  uniformly sampling from $\mathcal R_{\boldsymbol{x}}(\theta)$ is easy (using Algorithm 1 of \cite{jacob2019gibbs}), but uniform sampling from the union is much more difficult. The key observations in the paper are Proposition 3.1 and Proposition 3.2. They directly characterize the conditional distribution of $\boldsymbol{u}$ and make it possible for implementing the Gibbs sampler. We find the algorithm elegant and insightful.

\section{The math problem}

The new algorithm is interesting as a mathematics problem. The authors have translated it into a clever probability problem in the case of two categories. This gives the Markov chain:
\begin{align}\label{eqn:evolve}
Z^{(t)} = B_1^{(t)} (1 - B_2^{(t)}) Z^{(t-1)} + B_2^{t},
\end{align}
where $B_1^{t} \sim \Beta(N_1, 1)$ and $B_2^{t} \sim \Beta(1, N_2)$ are independent Beta random variables and $N_1, N_2$ are two fixed positive integers.

The chain falls into the well-studied area of iterated random functions  \cite{chamayou1991explicit}\cite{diaconis1999iterated}. It has been previously studied in \cite{letac2002donkey} who offer higher dimensional versions which might be relevant to the present paper when $k\geq 3$.

In the $k= 2$ case, let $P$ be the Markov transition kernel of the chain  and $\pi$ be its stationary distribution. The authors use coupling techniques to derive the following convergence  bound in terms of the Wasserstein-1 distance 
\begin{align}\label{eqn: was upper bound}
\cW_1 (P^t(z, \cdot), \pi) \leq \bigg(\frac{N_1}{N_1 + 1} \cdot \frac{N_2}{N_2 + 1}\bigg)^t \cdot W_1 (\delta_{z}, \pi)
\end{align}
for every initialization $z$, where $\delta_z$ is the delta-mass at  $z$. As a small contribution to the conversation, we derive a lower bound of the convergence speed. Our results suggest $\frac{N_1}{N_1 + 1} \cdot \frac{N_2}{N_2 + 1}$ is the exact convergence rate for the Markov chain. To start with, we prove the stationary distribution of the chain is another Beta distribution (this result can also be derived directly from the original formulation of the problem, see Appendix C of \cite{jacob2019gibbs}). The following lemma is well known and helpful: 
\begin{lem}\label{lem: beta distribution}
	Let $X\sim \Beta(a,b)$, $Y\sim \Beta(a+b, c)$ be independent random variables with $a,b,c >0$. Then $XY\sim \Beta(a, b+c)$.
\end{lem}

\begin{proof}
	We show the $k$-th moment of $XY$ equals the $k$-th moment of a $\Beta(a, b+c)$ random variable for every $k$.  Since both $XY$ and beta distributions are bounded by $[0,1]$, standard results in probability show that the distribution is characterized by all its moments. 
	
	The $k$-th moment for a $\Beta(a_1, a_2)$ random variable is $\frac{\Gamma(a_1 + k)}{\Gamma(a_1)} \frac{\Gamma(a_1 + a_2)}{\Gamma(a_1 + a_2 + k)}$, therefore,
	\begin{align*}
	\bE(XY)^k &= \frac{\Gamma(a + k)}{\Gamma(a)} \frac{\Gamma(a + b)}{\Gamma(a + b + k)} \times \frac{\Gamma(a + b+ k)}{\Gamma(a + b)} \frac{\Gamma(a + b + c)}{\Gamma(a + b + c + k)} \\
	& =  \frac{\Gamma(a + k)}{\Gamma(a)}  \frac{\Gamma(a + b + c)}{\Gamma(a + b + c + k)},
	\end{align*}
	which is the same as the $k$-th moment of $\Beta(a+b,c)$ distribution, as desired. 
\end{proof}

Now we are ready to show $\Beta(N_1+1, N_2)$ is stable under the transformation $Z \rightarrow B_1 (1 - B_2) Z + B_2$, and is thus the stationary distribution of the chain $\{Z^{(0)}, Z^{(1)}, \cdots, Z^{(n)}\}$. 
\begin{prop}
	Let $Z\sim \Beta(N_1 + 1, N_2), B_1 \sim \Beta(N_1, 1), B_2 \sim \Beta(1, N_2)$ be independent random variables, then $B_1 (1 - B_2) Z + B_2 \sim \Beta(N_1 + 1, N_2)$.
\end{prop}
\begin{proof}
	Let $Y := B_1 (1 - B_2) Z + B_2$, then  $1 - Y = (1-B_2) (1 - B_1Z)$. Lemma \ref{lem: beta distribution} shows $B_1 Z$ is a $\Beta(N_1, N_2 + 1)$ random variable, which in turn shows $1 - B_1 Z \sim \Beta(N_2+1, N_1)$. Since $1 - B_2 \sim \Beta(N_2, 1)$, apply Lemma \ref{lem: beta distribution} again yields $1- Y \sim \Beta(N_2, N_1 + 1)$. Therefore $Y$ is a $\Beta(N_1 + 1, N_2)$ random variable. 
\end{proof}

Now we turn to lower bound the convergence rate of the Markov chain. Recall that for any two probability measure $\mu, \nu$ with bounded support , the Kantorovich- Rubinstein dual theorem  shows:
\[
\cW_1(\mu, \nu)  = \sup_{\{f: \text{Lip}(f)\leq 1\}} \int f (d\mu - d\nu)
\]
Choosing $f(x)  = \pm x$ immediately implies $\cW_1(\mu, \nu) \geq \lvert m_1(\mu) - m_1(\nu)\rvert$, where $m_1$ stands for the first moment. Let $\nu = \pi$ and $Z^{(t)}(z)$ be the Markov chain at time $t$ which starts at $z$. It is clear that $m_1(\nu) = \frac{N_1 + 1}{N_1 + N_2 + 1}$, and the evolution equation \ref{eqn:evolve} yields
\begin{align}\label{eqn:recursion}
\bE(Z^t(z)) = \frac{N_1 N_2}{(N_1 + 1)(N_2 + 1)}  \bE(Z^{t-1}(z)) + \frac{1}{N_2 + 1}
\end{align}
Using the initial condition $\bE(Z^0(z)) = z $, recursion \ref{eqn:recursion} can be solved as:
\[
\bE(Z^t(z))  = \bigg( \frac{N_1 N_2}{(N_1 + 1)(N_2 + 1)}\bigg)^t \bigg(z - \frac{N_1 + 1}{N_1 + N_2 + 1}\bigg) + \frac{N_1 + 1}{N_1 + N_2 + 1}.
\] 
Thus we immediately have
\begin{align}\label{eqn: was lower bound}
\cW_1(P^t(z,\cdot),\pi) \geq \bigg\lvert z - \frac{N_1 + 1}{N_1 + N_2 + 1}\bigg \rvert \bigg(\frac{N_1 + 1}{N_1 + N_2 + 1}\bigg)^t.
\end{align}
Therefore, unless the chain is not initialized at exactly $\frac{N_1 + 1}{N_1 + N_2 + 1}$, the exact convergence rate of the Markov chain equals $\frac{N_1 + 1}{N_1 + N_2 + 1}$.
Combining \ref{eqn: was upper bound} and \ref{eqn: was lower bound}, we have the following:
\begin{prop}
	\begin{align}\label{eqn: was convergence rate}
  \lvert\bE Z - z \rvert	\leq \cW_1 (P^t(z, \cdot), \pi)\bigg/ \bigg(\frac{N_1}{N_1 + 1} \cdot \frac{N_2}{N_2 + 1}\bigg)^t \leq \bE\lvert Z - z \rvert,
	\end{align}
	where $Z\sim \Beta(N_1 + 1, N_2)$.
\end{prop}
When $z = 0$ or $1$, Formula \ref{eqn: was convergence rate} shows the convergence speed can be calculately exactly as the lower bound matches the upper bound perfectly. Moreover, the function $u: z \rightarrow \bE|Z - z| = \int_{0}^z \mathbb P (Z \leq t) dt + \int_{z}^1 \mathbb P (Z \geq t) dt$ has derivative $u'(z) = \mathbb P(Z \leq z) -\mathbb P(Z \geq z)$ almost everywhere. Therefore $u$ is first decreasing and then increasing on the unit interval, thus $\max_{z\in[0,1]} u(z) = \max\{u(0), u(1)\}$. This gives us the exact convergence speed for the chain under the worst-case sceneario:  
\begin{prop}
\[\sup_{z\in[0,1]}\cW_1 (P^t(z, \cdot), \pi) = \frac{\max\{N_1+1, N_2\}}{N_1 + N_2 + 1} \bigg(\frac{N_1 N_2 }{(N_1 +1)(N_2 + 1)}\bigg)^t.\]
\end{prop}
\begin{proof}
	Since $\bE|Z - z|$ is maximized at either $z = 1$ or $z= 0$, the RHS of Formula $\ref{eqn: was convergence rate}$ is upper bounded by  $\frac{\max\{N_1+1, N_2\}}{N_1 + N_2 + 1}$. Meanwhile, no matter $z = 1$ or $z= 0$, we know $|\bE Z- z| = \bE|Z - z|$, and thus we conclude
	\[\sup_{z\in[0,1]}\cW_1 (P^t(z, \cdot), \pi) = \frac{\max\{N_1+1, N_2\}}{N_1 + N_2 + 1} \bigg(\frac{N_1 N_2 }{(N_1 +1)(N_2 + 1)}\bigg)^t.\]
\end{proof}

\begin{q}
Can the authors see if their general algorithm can be translated into a vector-valued version of \ref{eqn:evolve} for $k$ greater than or equal to $3$?
\end{q}

\section{Hypothesis Testing Problems}
The introduction to the paper emphasizes problems where the sample size $N$ is small compared to the number of categories. We have encountered such problems in our recent work and we find straightforward Bayesian and frequentist solutions.

\begin{eg}
 In studying the popular `wash shuffle' where a deck of $52$ cards is `smushed around the table for $t$ seconds' (say $t = 60$) one wants to test if the cards are well mixed. Here the data consists of $100$ permutations of $X_1, X_2, \cdots X_{100}$ in $S_{52}$. The number of categories $k = \lvert S_{52}\rvert \approx 8.6 \times 10^{68}$ is huge. One has some partial prior information. For example, if the cards are not mixed, it may be because several cards original together are still together. Or the original top (or bottom) card may still be close to the top (or bottom). In our analysis, we found statistics $T_1, T_2, \cdots T_l$ (e.g., $T_1 =$ the number of adjacent pairs, $T_2 =$ position of the original top card, $T_3 = $ the distance from the starting distribution in some natural metric, $\dots$).
 
 The distribution of these features can be obtained under the null distribution. This allows standard frequentist tests. In our work \cite{diaconis2018bayesian}, we provided a Bayesian solution to this problem. We considered an exponential family through the statistics. For $\sigma \in S_{52}$:
 \[
 \bP_\theta(\sigma) = Z^{-1}(\sigma) \exp(\sum_{i=1}^l \theta_i T_i(\sigma)),
 \]
 where $Z(\sigma) = \sum\limits_{\sigma\in S_{52}} \exp(\sum_{i=1}^l \theta_i T_i(\sigma))$ is the normalization constant.  Thus $\theta = 0$ is the uniform distribution. The exchange algorithm \cite{murray2006mcmc} allowed us to compute the posterior in a reasoning fashion. 
\end{eg} 

\begin{q}
	Do the authors think there will be a time for their algorithm can be run for such $N$ and $k$? 
\end{q}

\section{Remark}
As clearly explained in the paper, the uniform distribution on the $k$-simplex underlies the basic procedure. It is well to be reminded that it is hard to understand the properties when $k$ is large. The following two examples are drawn from \cite{diaconis2002bayesian}.
\begin{eg}[Bayesian Birthday Problem]
	Consider $N$ balls dropped into $k$ boxes, with probability Multinomial$(N,\boldsymbol{\theta})$. What is the chance that all the balls are dropped into distinct boxes?
	
	If $k = 365$, this is the classical birthday problem. The classical frequentist answer takes $\boldsymbol{\theta} = (\frac1{365}, \cdots, \frac 1{365})$ and shows that the chance is approximately $0.5$ when $N = 23$.  A `flat prior Bayesian' puts a uniform prior on the $k$-simplex. Then the chance of all balls in distinct boxes  is approximately $0.5$ when $N = 16$. Here, being a Bayesian does not change things much.
\end{eg}

\begin{eg}[Bayesian Coupon Collectors Problem]
With notation as above, consider the question: How large does $N$ have to be so that the chance that all boxes are covered is close to $0.5$?

The frequentist solution assumes $\boldsymbol{\theta} = (\frac1{365}, \cdots, \frac 1{365})$ and then shows $N$ should be approximately $2287$. Using a uniform prior for $\boldsymbol{\theta}$, one finds $N$ has to be $191844$. Here, the uniform prior makes a huge difference. 
\end{eg}

\begin{q}
	Is the uniform distribution on the simplex a crucial part of the procedure or could this be varied?
\end{q}

\section{A Bit of History}
When P.D. was a beginning graduate student at Harvard (1971), Art Dempster called him in to suggest a possible thesis topic: `Find ways to do the computations required for upper and lower probabilities'. This did not work out at the time but triggered a lifetime's interest. It is inspiring to have tracked his efforts over a $50$ year period. The present paper is important progress. We look forward to progress on the problems mentioned in our introduction.

\bibliographystyle{alpha}
\bibliography{refs}{}
\end{document}